\documentclass[11pt,a4paper]{article}

\usepackage[margin=3.5cm]{geometry}
\usepackage{here}
\usepackage[all]{xy}

\usepackage[ruled,vlined,linesnumbered]{algorithm2e}

\usepackage{hyperref}

\usepackage{mymacro}
\usepackage{url}

\usepackage{comment}

\newcommand{\R}{\mathbb{R}}

\makeatletter

\@addtoreset{equation}{section}
\makeatother

\title{Curvature of point clouds through \\ principal component analysis}
\author{Yasuhiko Asao\footnote{Fukuoka University, 8-19-1 Nanakuma, Jonan-ku, Fukuoka city, Fukuoka, Japan. \texttt{asao@fukuoka-u.ac.jp}} \and Yuichi Ike\footnote{Graduate School of Information Science and Technology, The University of Tokyo, 7-3-1 Hongo, Bunkyo-ku, Tokyo 113-8656, Japan. \texttt{yuichi.ike.1990@gmail.com}, \texttt{ike@mist.i.u-tokyo.ac.jp}} 
}
\date{\today}

\begin{document}

\maketitle

\begin{abstract}
In this article, we study curvature-like feature value of data sets in Euclidean spaces. First, we formulate such curvature functions with desirable properties under the manifold hypothesis. Then we make a test property for the validity of the curvature function by the law of large numbers, and check it for the function we construct by numerical experiments. These experiments also suggest the conjecture that the mean of the curvature of sample manifolds coincides with the curvature of the mean manifold. Our construction is based on the dimension estimation by the principal component analysis and the Gaussian curvature of hypersurfaces. 
Our function depends on provisional parameters $\varepsilon, \delta$, and we suggest dealing with the resulting functions as a function of these parameters to get some robustness.
As an application, we propose a method to decompose data sets into some parts reflecting local structure. 
For this, we embed the data sets into higher dimensional Euclidean space using curvature values and cluster them in the embedding space.
We also give some computational experiments that support the effectiveness of our methods.
\end{abstract}


\section{Introduction}

\subsection{Motivation}

Most of data we deal with take the form of point clouds in Euclidean spaces. 
As many data analysts propose, they are often distributed around an embedded manifold in the Euclidean space, which is now called the \emph{manifold hypothesis} (see, for example, \cite{ma2011manifold,NSW,fefferman2016testing}). 
To extract some feature values from them is a pivotal task for data analysis. 
From the standpoint of the manifold hypothesis, estimating its curvature seems effective and natural. 
Several authors study the curvature of data from the view point of manifold learning or 3D shape study.
See \cref{subsec:related} for these related works. 

In this article, we first formulate a curvature function of data to satisfy some desirable properties under the manifold hypothesis, and we construct a candidate for such a function. Then we make a test property for the validity of the curvature function based on the law of large numbers, and check it for the function we construct. Our construction is based on the dimension estimation by the principal component analysis and the Gaussian curvature of hypersurfaces. 
We also give a computation algorithm for it as outlined in the next subsection.

\subsection{Summary of our construction}\label{summary}

Suppose that the data set $X$ we deal with is a finite subset of the Euclidean space $\R^{n}$ for some $n$. 
We estimate the curvature of $X$ at $p \in X$ by the following procedure. 
\begin{enumerate}
\item Fix thresholds $\varepsilon > 0, \delta > 0$. 
\item Calculate the eigenvalues $\lambda_{1} \ge \cdots \ge \lambda_{n}$ of the covariance matrix of the $\varepsilon$-neighborhood of $p$. 
Let $u_{1}, \cdots, u_{n}$ be the corresponding eigenvectors. 
\item Look at the eigenvalues $\lambda_{1} \ge \cdots \ge \lambda_{K} \geq \delta$ that are greater than the threshold $\delta$. 
We regard the affine space $p+\langle u_{1}, \cdots, u_{K}\rangle$
as the tangent space of the data at the point $p$, and also regard $u_{K+1}$ as the normal direction at $p$.
\item We give a frame and an orientation of the tangent space. 
We project all the points around $p$ to the affine space spanned by the tangent space and the normal direction. 
Then we fit the resulting points by quadratic hypersurface in the affine space, and compute its Gaussian curvature at $p$ with respect to the fixed frame.
\end{enumerate}
The steps (i)--(iii) are essentially the same as principal component analysis. 
See \cref{subsec:tangent} for the mathematical background of (i)--(iii), and see \cref{prp:aprx} for that of (iv). 
For the detail of the algorithm, see \cref{sec:alg}.

\subsection{Contributions}

Our contributions are the following.
\begin{enumerate}
    \item[(1)] In \cref{sec:framework}, we formulate curvature functions of data under the manifold hypothesis probability theoretically. Based on this formulation, we make a test property which is desirable for such a function. This property is verifiable by computations.
    \item[(2)] In \cref{sec:theoretical}, we define the ``curvature" of point sets and give theoretical foundation that supports the validity of the definition. 
    For that purpose, we investigate a modified version of the principal component analysis and an estimation method of Gaussian curvature.
    \item[(3)] In \cref{sec:alg}, we propose a novel algorithm to compute the curvature based on the theoretical background. 
    We also propose a clustering algorithm relying on the curvature values of a point set.
    \item[(4)] In \cref{sec:exp}, we experimentally show the effectiveness of our method using artificial and real-world data sets. 
\end{enumerate}

\subsection{Related work}\label{subsec:related}

Our method first computes the tangent space at each point, which is related to studies in manifold learning (see, for example, \cite{lee2007nonlinear,ma2011manifold,wang2012geometric}). 
There one usually use $k$-nearest neighbor graph instead of a ball centered at each point as in our algorithm. 
Locally Linear Embedding (LLE)~\cite{roweis2000nonlinear} estimates tangent spaces of a point set relying on $k$-nearest neighbor graph.
\cite{bangio2004non-local} proposes a method to estimate tangent spaces by using neural networks whose input is global information of a point set, while LLE only uses local information. 
In most of the studies in manifold learning, it is often sufficient to estimate the tangent space for machine learning tasks, hence there are few studies on non-linear approximation of the manifold as in our method.

In 3D shape studies, curvature-like feature values have been studied from various point of views.
\cite{taubin1995estimating} proposes estimating curvature via polyhedral approximation. 
\cite{lange2005anisotropic} computes curvature of a point set by solving some minimization problems to estimate tangent spaces.
\cite{cao2019efficient} studies the Weingarten map from a point set and estimates its convergence rate.
In this article, we do not only focus on 3D shapes, but also explore a method that can be applied to broader type of data satisfying the manifold hypothesis.

A mathematical framework of ours related to the manifold hypothesis is close to one considered in \cite{NSW}. They formalize the hypothesis by considering a distribution on the normal bundle of a submanifold of $\R^n$, where the submanifold is conceptualized as geometric core of the distribution, and the normal direction is considered for noise distribution. On the other hand, our formulation does not have any restriction on the direction of noise. Another formulation is discussed in \cite{fefferman2016testing}, where they consider a distribution supported on a unit sphere in a Hilbert space.
In \cite{niyogi2008finding,NSW}, the authors 
discuss how to recover the homology from observations by utilizing a quantity $1/\tau$ which is closely related to the curvature of manifolds. 
The relationship between $1/\tau$ and the curvature function we construct is not clear and we plan to investigate it in future work.
In \cite{AKTW}, they study the limit behavior of $1/\tau$ as the dimension of ambient space gets large.
Subsequent research such as \cite{fasy2014confidence,chazal2017robust} studies the problem combining methods in topological data analysis with statistics.
Our theoretical setting is somewhat similar to theirs (see \cref{rem:nsw}), but our method focuses on geometrical features rather than topological ones.

\subsection*{Acknowledgements}

The first author was supported by RIKEN Center for Advanced Intelligence Project.
The second author was partially supported by JST ACT-X (JPMJAX1903).

\section{Our framework involving manifold hypothesis}\label{sec:framework}

In this section, we explain the hypothetical framework under which we work.

Let $\Omega$ be a probability space. 
We consider a stochastic process $\Theta \colon \Omega \times \R^m \to \R^n$ centered at a manifold, that is, the restriction $\Theta(-, p) \colon \Omega \to \R^n$ is measurable for every $p \in \R^m$, and the map $\int \Theta \colon \R^m \to \R^n$ defined by 
\begin{equation}
    \int \Theta (p) = \int \Theta(\omega, p)d\omega
\end{equation}
is an embedding. 
We call such a stochastic process \emph{a stochastic embedding}. We regard stochastic embeddings as a source of data obeying the manifold hypothesis. We will later formulate data as some subset of the image of samples of stochastic embeddings. 

\begin{remark}\label{rem:nsw}
If we have a restriction that $\Theta(\omega, p)$ varies (with respect to $\omega$) only in the normal direction of the mean manifold at $p$, this formulation is very close to that of Niyogi--Smale--Weinberger's~\cite{NSW}. 
They consider that the distribution of data is induced from that on the normal bundle of a manifold. One of the merits of our approach using stochastic process is that we can deal with multiple points simultaneously which are generated with some noise from distinct points of the manifold. This is convenient to consider the curvature of the sample points. On the other hand, Niyogi--Smale--Weinberger consider a quantity $1/\tau$ associated to an embedded manifold that is a curvature-like quantity. However, it seems difficult to estimate $1/\tau$ from the data. 
\end{remark}

We denote by $\Emb(\R^m, \R^n)$ the set of embeddings $\R^m \to \R^n$, and also denote by $\Emb(\Omega \times \R^m, \R^n)$ the set of stochastic embeddings $\Omega \times \R^m \to \R^n$. 
We also denote by $\SP(\Omega \times \R^m, \R^n)$ the set of stochastic processes $\Omega \times \R^m \to \R^n$, that is a collection of maps $\Omega \times \R^m \to \R^n$ satisfying that $\Theta(-, p) \colon \Omega \to \R^n$ is measurable for every $p \in \R^m$. 
Assume that we have a map $c \colon \Emb(\R^m, \R^n) \to \Map(\R^m, \R)$ that extracts geometric feature of each point of embeddings, and we call such a map \emph{curvature function}. 
Here $\Map(A,B)$ denotes the set of maps from $A$ to $B$.
For example, we can take $c$ as the sectional curvature of the embedded manifold equipped with a metric induced from the standard metric of the ambient Euclidean space. 
We further assume that, corresponding to a curvature function $c$, there is a map $c_\Omega : \Emb(\Omega \times \R^m, \R^n) \to \SP(\Omega \times \R^m, \R)$ that plays as a stochastic version of curvature function. We may expect that $c_\Omega$ is a `blurring' of the curvature map $c$, and  the expectation of $c_\Omega$ coincides with $c$. That is, the maps $c$ and $c_\Omega$ are expected to satisfy the following commutative diagram:
\begin{equation}\label{diagram1}
\begin{split}
\xymatrix@C=30pt{
\Emb(\Omega \times \R^m, \R^n) \ar[r]^-{c_\Omega} \ar[d]_-{\int} & \SP(\Omega \times \R^m, \R) \ar[d]^-{\int} \\
\Emb(\R^m, \R^n) \ar[r]_-{c} & \Map(\R^m, \R), 
}
\end{split}
\end{equation}
where $\int$ denotes the integration over $\Omega$. 
We do not prove the existence of such a map $c_\Omega$ (it will be studied in our future work), but we can check the existence if we restrict $\Emb(\Omega \times \R^m, \R^n)$ to the set of stochastic embeddings with the following conditions:
\begin{enumerate}
    \item $m = 2, n = 3$, and $\int \Theta$ is a surface,
    \item $\Theta(\omega, -) \colon \R^m \to \R^n$ is smooth for every $\omega \in \Omega$,
    \item $\partial_{ij}\Theta(-, p) \colon \Omega \to \R^n$, the second derivatives at $p \in \R^m$, are uncorrelated for every $p$. That is, the covariance of each pair of the differentials is 0.
\end{enumerate}
If $\Theta$ satisfies the above conditions, then the mean of the Gaussian curvature coincides with the Gaussian curvature of the mean manifold $\int \Theta$. 
We consider that such a restriction makes sense when we deal with real-world data. For example, a Gaussian process $\Theta : \Omega \times \R^2 \longrightarrow \R^3$ with $\Theta(\omega, p) = (p, t)$, $t \sim N(0, 1)$ whose covariance is determined the kernel $K(p, q) = p^{\mathrm T}q$ satisfies the above. It is checked by calculating 
\[
\int \partial_{ij}\Theta(\omega, p)\partial_{i'j'}\Theta(\omega, p)d\omega = \partial_{ij}\partial_{i'j'}K(p, q)|_{p=q} = 0.
\]
For another example, we can take the RBF kernel $K(p, q) = e^{-\|p-q\|^2}$ for the same process so that mean of the curvature coincides with the curvature of the mean. We can check it by calculating
\[
\int \partial_{11}\Theta(\omega, p)\partial_{22}\Theta(\omega, p)d\omega = \partial_{11}\partial_{22}K(p, q)|_{p=q} = 4,
\]
and
\[
\int \partial_{12}\Theta(\omega, p)\partial_{12}\Theta(\omega, p)d\omega = \partial_{12}\partial_{12}K(p, q)|_{p=q} = 4.
\]
We also check it numerically in \cref{subsec:noise}.

Next we formulate the term ``\emph{data}" under the manifold hypothesis. 
Let $X$ be a finite subset of $\R^m$.
For a fixed $\omega \in \Omega$, we define a map $(\omega, X) \colon \Emb(\Omega \times \R^m, \R^n) \to \Map(X, \R^n)$ by 
\[
(\omega, X)\Theta = \Theta(\omega, -)|_{X}.
\]
We define \emph{a data} as an element in $\Map(X, \R^n)$ which is contained in the image of $(\omega, X)$ for some $\omega$. For a data $D \in \Map(X, \R^n)$, we call each element $D(x)$ for $x \in X$ \emph{a datum}.
We denote $\Data_{\Omega}(X, \R^n)$ the set of data. A pivotal task of data analysis is to extract some feature value from a data which is a finite point cloud. Hence we formulate such task including curvature estimation as a map $\Data_{\Omega}(X, \R^n) \to \Map(X, \R)$. In contrast to such maps, we regard $c_\Omega$ as an `ideal feature value function', and we expect that extraction from a data is an approximation of that by ideal function. Hence we define \emph{a curvature function of data} $c_d : {\sf Data}_{\Omega}(X, \R^n) \to \Map(X, \R)$ as a function with the following commutative diagram:
\begin{equation}\label{diagram2}
\xymatrix{
\Emb(\Omega \times \R^m, \R^n) \ar[rr]^-{\int \circ c_\Omega} \ar[d]_-{X^\ast} && \Map(\R^m, \R) \ar[d]^-{X^\ast} \\
\Map(\Omega \times X, \R^n) \ar[r]_-{\tilde{c}_d} &\Map(\Omega, \Map(X, \R)) \ar[r]_-{\int} & \Map(X, \R), 
}
\end{equation}
where $X^\ast$ denotes the restriction to $X$, $\tilde{c}_{d}$ is only defined on ${\rm Im} X^\ast$ by $\tilde{c}_d(X^\ast \Theta)(\omega) = c_{d}((\omega, X)\Theta)$, and the map $\int $ in the bottom is defined by $\int \varphi (x) = \int \varphi(x)(\omega)d\omega$ for any $\varphi \in \Map(\Omega, \Map(X, \R))$ and $x \in X$.  
If such a map $c_d$ exists, then we have the convergence
\begin{equation}\label{eq:lln}
\sum_i \frac{c_d(\omega)\circ X^\ast \Theta_i}{n} \to \int c_d(\omega)X^\ast \Theta_i d\omega = c|_{X}\int \Theta_i,
\end{equation}
where $\Theta_i$'s in $\Emb(\Omega \times \R^m, \R^n) $ by the law of large numbers. 
In this article, we construct a candidate of the curvature function of data $c_d$ and show its validity by checking \eqref{eq:lln} in \cref{subsec:noise}.

\section{Theoretical background}
\label{sec:theoretical}

In this section, we give the theoretical background that supports our method. 
In the first subsection, we review \emph{principal component analysis} with slight modification according to our usage. 
For the usual principal component analysis, we refer to, for example, \cite{jolliffe1986principal,abdi2010principal}. 
This subsection details the theoretical background of (i)--(iii) in \cref{summary}. 
In the second subsection, we explain how to estimate curvature of point sets from Gaussian curvature of hypersurface by using principal components. 
This subsection details the background of  (iv) in \cref{summary}.

\subsection{Our usage of principal component analysis}\label{subsec:tangent}

Let $X = \{p_1, \dots, p_N \} \subset \R^{n}$ be a finite subset and $p \in \R^n$. 
For an inner product space $(W, \cdot \ )$, we denote by $S(W)$ the set of all unit vectors in $W$. 
Let $u, v \in S(\R^n)$. 
We define \emph{the covariance of $X$ from $p$ along $u$ and $v$} by 
\[
V_X(p, u, v) := \frac{1}{N}\sum_{i=1}^N \left((p - p_i)\cdot u \right)\left((p - p_i)\cdot v \right).
\]
We write $V_X(p, v, v) = V_X(p, v)$ and call it \emph{the variance of $X$ from $p$ along $v$}. 
The variance $V_X(p, v)$ indicates how the data $X$ is scattered in the direction $v$ from the point $p$. 
We define the \emph{covariance matrix of $X$ from $p$} by $V_X(p) := (V_X(p, e_i, e_j))_{i, j} \in M_n(\R)$, where $e_{i}$'s are the standard basis of $\R^n$. 
Using $V_X(p)$, we can write $V_X(p, v)$ as 
\begin{align}
V_X(p, v) &= \frac{1}{N}\sum_{i=1}^{N} \left((p - p_i)\cdot v \right)^2 \\
&= \frac{1}{N}\sum_{i=1}^{N} v^T (p - p_i)(p - p_i)^T v \\
&= v^T \left( \frac{1}{N}\sum_{i=1}^{N} (p - p_i)(p - p_i)^T \right) v \\
&= v^T V_X(p)v.
\end{align}
Hence all the eigenvalues of $V_X(p)$ are non-negative since $V_X(p, v) \geq 0$ for all $v$. 
The \emph{principal component analysis} method gives an orthogonal decomposition $\R^n \cong W_1\oplus \cdots \oplus W_m$ such that any unit vector in $S(W_i)$ maximizes $V_X(p, v)$ under the restriction  $v \in S(W_{i}\oplus \cdots \oplus W_m)$.
More precisely, we have the following proposition.

\begin{proposition}\label{prp:pca}
Let $\lambda_1 >  \dots > \lambda_m \ (1 \leq m \leq n)$ be the eigenvalues of the covariance matrix $V_X(p)$, and let $W_1, \dots, W_m$ be the corresponding eigenspaces, respectively. 
Then for $1 \leq i \leq m$ and any $v \in S(W_i)$, 
\[
v = \argmax_{u \in S(W_{i}\oplus \cdots \oplus W_m)} V_{X}(p, u).
\]
\end{proposition}

\begin{proof}
For $u_i \in W_i$ and $v_j \in W_j$ with $i\neq j$, we have $u_i \cdot v_j = 0$ since $\lambda_i(u_i \cdot v_j) = u_{i}^{T}V_X(p)v_j = \lambda_j(u_i \cdot v_j)$. 
Hence $W_i$'s are orthogonal to each other. 
By taking an orthonormal basis of each $W_i$'s, we have an orthonormal basis $u_1, \dots, u_n$ of $W_1\oplus \cdots \oplus W_m$. 
We denote the eigenvalue corresponding to $u_i$ by $\lambda'_i$. 
Obviously, we have $\lambda'_1 \geq \dots \geq \lambda'_n$. 
Let $v = \sum_{j = k}^{n} a_{j}u_j \in S(W_i\oplus \cdots \oplus W_n)$. 
Since $u_j$'s are orthonormal, we have $\sum_{j = k}^{n} a_{j}^2 = 1$. Then we have 
\begin{align}
V_{X}(p, v) &= v^{T}V_{X}(p)v\\
&= \left(\sum_{j = k}^{n} a_{j}u_j \right)^{T }V_{X}(p) \left( \sum_{j = k}^{n} a_{j}u_j \right) \\
&= \left( \sum_{j = k}^{n} a_{j}u_j \right)^{T} \left( \sum_{j = k}^{n} \lambda'_{j}a_{j}u_j \right) \\
&= \sum_{j = k}^{n} \lambda'_{j}a_{j}^2\leq \lambda'_k \leq \lambda_i.
\end{align}
Hence, for any $v \in S(W_i\oplus \cdots W_n)$, we obtain $V_{X}(p, v) \leq \lambda_i$. 
On the other hand, for $v \in S(W_i)$, we have $V_{X}(p, v) = v^{T}V_{X}(p)v = \lambda_{i}v^{T}v = \lambda_i$.
This completes the proof.
\end{proof}

Let us introduce a terminology for later use. 
Consider an embedding of a manifold $M \to \R^n$. 
If it factors through an embedding $\iota \colon M \to W$ for some affine subspace $W \subset \R^n$, we call the embedding $\iota$ an \emph{efficient embedding}. 
From the standpoint of the manifold hypothesis that every point set we consider is distributed around an embedded manifold, the variance from a point on the manifold in the direction of its tangent vector should be large, while it should be very small in the other direction normal to the tangent space. 
Hence if we decompose $\R^{n}_{p}$, the set of all vectors from $p$, as $\R^{n}_{p} \cong W_1\oplus \cdots \oplus W_m$ by \cref{prp:pca}, it is reasonable to regard its subspace $W_1\oplus \cdots \oplus W_k$ as the tangent space at $p$. 
Here we fix some threshold $\delta>0$, and $k$ is determined by the conditions $\lambda_k \geq \delta$ and $\lambda_{k+1} < \delta$.
By adding an eigenvector in $W_{k+1}$ to this affine space, we regard it an ambient space in which the underlying manifold is efficiently embedded.

\subsection{Estimation of Gaussian curvature}\label{subsec:gauss}

To estimate the curvature of a point set, we use the notion of Gaussian curvature of hypersurfaces. 
See \cite{petersen2006riemannian} for Gaussian curvature of Riemannian manifolds.
It is known that if a hypersurface in $\mathbb{R}^{n}$ is expressed by an explicit function as $x_{n} = f(x_{1}, \cdots, x_{n-1})$ with $\partial f/\partial x_{i} = 0$ for $1 \leq i \leq n-1$, then its Gaussian curvature at the origin is expressed by the determinant of the Hessian $H_{f} = (\partial^2 f/\partial x_{i}\partial x_{j})_{i, j}|_{0}$. 
If a manifold is efficiently embedded in codimension 1, then we can see the manifold as a hypersurface and compute its Gaussian curvature. 
Let $\iota \colon M \to \R^n$ be an embedding of $d$-manifold, and $a \in M$. Even if $\iota$ is not codimension 1, under the assumption that the linear space $W$ spanned by the first and the second derivatives of $\iota$ at $a$ is $(d+1)$-dimensional, the following proposition guarantees that the orthogonal projection to $W$ sufficiently approximates its Gaussian curvature at $a$ by regarding the projeted manifold as an embbeding in $W$. 
Here we estimate the Gaussian curvature at $a$ by regarding the direction in $W$ normal to the tangent space of $\iota(a)$ as the codimensional direction.

\begin{proposition}\label{prp:aprx}
Let $\iota \colon \mathbb{R}^{m} \to \mathbb{R}^{n}$ be an embedding of class $C^2$, and let $a \in \mathbb{R}^m$. Consider the $\mathbb{R}$-vector space $W_{a}$ spanned by 
\begin{align}\label{vectors}
\left\{\frac{\partial \iota}{\partial x_{1}}(a), \dots, \frac{\partial \iota}{\partial x_{m}}(a), \frac{\partial^2 \iota}{\partial x_{1}\partial x_{1}}(a), \frac{\partial^2 \iota}{\partial x_{1}\partial x_{2}}(a), \dots, \frac{\partial^2 \iota}{\partial x_{m}\partial x_{m}}(a)\right\}.
\end{align}
Then the affine space $\iota(a) + W_{a} \subset \mathbb{R}^n$ has the following property. 
Let $P \colon \mathbb{R}^n \to \mathbb{R}^n$ be the projection to $\iota(a) + W_{a}$. 
Then for any $\varepsilon > 0$, there exist $\varepsilon' > 0$ such that the map $\iota|_{B(a; \varepsilon')}$ is $\varepsilon$-close to $P \circ \iota|_{B(a; \varepsilon')}$ in $C^{2}$, where $B(a; \varepsilon')$ denotes the $\varepsilon'$-ball centered at $a$.
\end{proposition}

\begin{proof}
We have $\iota|_{B(a; \varepsilon')}(a) = P\circ \iota|_{B(a; \varepsilon')}(a)$. 
Furthermore, since the affine space $\iota(a) + W_{a}$ is spanned by the vectors \eqref{vectors}, we have 
\begin{align}
\frac{\partial }{\partial x_{i}}P \circ \iota (a) &= P\left(\frac{\partial \iota}{\partial x_{i}}(a)\right) = \frac{\partial \iota}{\partial x_{i}}(a), \\
\frac{\partial^{2} }{\partial x_{i}\partial x_{j}}P \circ \iota (a) &= P\left(\frac{\partial^{2} \iota}{\partial x_{i}\partial x_{j}}(a)\right) = \frac{\partial^{2} \iota}{\partial x_{i}\partial x_{j}}(a)
\end{align}
for every $1 \leq i, j \leq m$. Hence the statement follows from the continuity of the map and the derivatives. 
\end{proof}

Now we provide our definition of \emph{``curvature" of embeddings and point sets}, which is not the standard notion of curvature, but enough for extracting a feature value of embeddings or point sets. 
For an embedding $\iota$, we choose a normal direction $n_a$ of tangent space at $\iota(a)$ for each point $\iota(a)$, and denote by $W_a$ the affine space spanned by the tangent space at $\iota(a)$ and $n_a$. 
We project the embedding $\iota$ to $W_a$ and regard its image around $\iota(a)$ as a hypersurface in $W_a$. 
We define the curvature of the embedding at $\iota(a)$ as the Gaussian curvature at $\iota(a)$. 
By \cref{prp:aprx}, this sufficiently approximates the Gaussian curvature in the situation explained right before the proposition. 
Let $X$ be a point set and $p \in X$. 
If they are distributed around an embedding of a manifold, it is easy to understand the following, but our procedure is irrelevant to the assumption. 
As explained in \cref{subsec:tangent}, we regard the affine space $W_p = W_1\oplus \cdots \oplus W_k$ determined by some fixed $\delta > 0$ as the tangent space of the underlying manifold. 
We choose an eigenvector in $W_{k+1}$ in the manner explained below, and regard it as the normal direction. 
We choose an orientation of the tangent space as explained later, and estimate the Gaussian curvature at $p$ by fitting a polynomial surface around $p$ regarded as a hypersurface. 
Note that, when the tangent space is even dimensional, the Gaussian curvature is independent of the choice of an orientation. 
It is checked by calculating the Hessian with respect to the orientation change.
The choice of the normal direction and the orientation is done as follows. 
Let $u_1, \dots u_K$ be orthonormal basis of $W_p = W_1\oplus \cdots \oplus W_k$ such that each $u_i$ belongs to some $W_j$. 
Such a basis exists since $W_j$'s are orthogonal to each other. 
See also the proof of \cref{prp:pca}.
Consider $u_{i}$'s as points on the sphere $S^{n-1}$. 
We give a decomposition $U \sqcup S^{n-1}\setminus U$ such that $v \in U$ or $-v \in U$ for every point $v \in S^{n-1}$, then we choose $u_{i}$'s so that they belong to $U$. 
Let $S^{i-1}_{+} = \{(x_{1}, \cdots, x_{i}) \in S^{i-1} \mid x_{i} > 0 )\}$ be the upper semi-sphere for $i \geq 0$. 
Inductively we define $U_{0} = \{ 1 \} \subset S^{0}$ and $U_{i} = S^{i}_{+} \cup U_{i-1}$ for $1\leq i \leq n-1$. 
Then the obtained $U_{n-1} \subset S^{n-1}$ satisfies the desired condition. We adjust the signatures of $u_i$'s so that they satisfy that $u_i \in U_{n-1}$. 
This determines a frame and an orientation of $W_p$, hence those of the tangent space at $p$.
It also determines $u_{K+1}$, which is the normal direction.

\section{Algorithm for computing curvatures of point clouds}\label{sec:alg}

In this section, we give our algorithm to compute the curvature of a given point set. 
We also provide a clustering algorithm of point sets depending on the curvature values. 
Until the end of this section, let $X$ be a finite set in $\bR^n$.

\subsection{Curvature of point clouds}\label{subsec:curv-alg}

Our aim is to calculate a curvature-like quantity for each point $p \in X$. 
In order to work locally near $p \in X$, we fix a threshold $\varepsilon>0$ and consider the subset $B:=X \cap B(p;\varepsilon)$ of $X$, where $B(p;\varepsilon)$ denotes the open ball of radius $\varepsilon$ centered at $p$.
First we apply the principal component analysis described in \cref{subsec:tangent} to $B$. 
That is, we diagonalize the covariance matrix $\Var_{B}(p)$ as $U^TDU$, where  $D=\diag(\lambda_1,\dots,\lambda_n)$ with $\lambda_1 \ge \dots \ge \lambda_n$ and $U$ is an orthogonal matrix.
Write $U=[u_1, \dots, u_n]$.
We give an orientation for $u_1,\dots,u_n$ as in \cref{subsec:gauss}. It is pursued for each $u_{i} = (s^{i}_{1}, \dots, s^{i}_{n})$ as follows. 
Look at the $n$-th coordinate $s^{i}_{n}$, and set $u_{i} = -u_{i}$ if $s^{i}_{n}<0$, and remain if $s^{i}_{n}>0$. 
If $s^{i}_{n}=0$, look at the $(n-1)$-th coordinate $s^{i}_{n-1}$ and repeat this procedure.

Now, set another threshold $\delta>0$ and let $\lambda_1 \ge \dots \ge \lambda_K \ge \delta$ be the eigenvalues at least $\delta$. 
We consider the $K$-dimensional affine space $p + \langle u_1,\dots, u_{K} \rangle$ that is centered at $p$ and spanned by vectors $u_1,\dots,u_K$. 
We regard $K$ as the dimension of $X$ around $p$, since the affine space would approximate the tangent space of the centered true manifold as explained in \cref{sec:theoretical}.
The pseudocode for computing the dimension $K$ and vectors $u_1, \dots u_K, u_{K+1}$ is shown in \cref{algorithm:dimension}.

Next, we shall compute the curvature of $X$ at $p$.
For a point $q \in \bR^n$, we set $x_i(q)= (q-p)\cdot u_i $ for $i=1,\dots,K+1$. 
We fit a quadratic hypersurface $x_{K+1} = \frac{1}{2} \sum_{i,j} a_{i,j} x_i x_j=:f(x_1,\dots,x_K)$ with $a_{i,j}=a_{j,i}$ to $X$.
For this, we minimize the square error 
\begin{equation}
    E(a)=\sum_{q \in B} \left( \frac{1}{2} \sum_{i,j}a_{i,j} x_i(q) x_j(q) -x_{K+1}(q) \right)^2.
\end{equation}
Solving the equations $\partial E/\partial a_{i,j}=0$ gives $a_{i,j}$'s. 
We define the curvature of $X$ at $p$ by the Hessian $H_f$ of $f$, which is computed as 
\begin{equation}
    H_f=
    \det 
    \begin{bmatrix}
      a_{1,1} & \cdots & a_{1.K} \\
      \vdots & \ddots & \vdots \\
      a_{K,1} & \cdots & a_{K,K}
    \end{bmatrix}.
\end{equation}
The pseudocode for computing the curvature is shown in \cref{algorithm:curvature}.

Combining the procedures described above, we obtain an algorithm to compute the dimension and the curvature of $X$ at $p$ given thresholds $\varepsilon, \delta>0$. 
The pseudocode for the overall algorithm is given in \cref{algorithm:dim-curv}.
We can use the parameter $\varepsilon$ as a persistent parameter as in topological data analysis, which encodes a multi-scale topological feature of a given point set.
That is, varying $\varepsilon$ from $0$ to $+\infty$ we track change of values of dimension and curvature. 
A small $\varepsilon$ reflects the local structure of $X$ near $p$ while a large reflects the global one. 
Moreover, if the value is close to a fixed constant for different $\varepsilon$, that is, if the value ``persists" for a long time, we can regard it as the intrinsic one.
In \cref{sec:exp}, we show change of dimension and curvature depending on $\varepsilon$ in various examples.

\begin{algorithm}[!ht]
\caption{Compute dimension}
 \label{algorithm:dimension}
 \SetKwInOut{Input}{input}
 \SetKwInOut{Output}{output}
 \SetFuncSty{textrm}
 \SetCommentSty{textrm}
 \SetKwFunction{ComputeDimension}{{\scshape ComputeDimension}}
 \SetKwProg{myfunc}{}{}{}
\Input{$X$: point set, $p \in X$: point, $\varepsilon, \delta$: positive numbers}
\Output{$\dim(X, p; \varepsilon, \delta), (u_1, \dots, u_{\dim(X, p; \varepsilon, \delta)+1})$: non-negative integer and orthonormal vectors}
\myfunc{\ComputeDimension{$X, p; \varepsilon, \delta$}}{
$B \gets X \cap B(p;\varepsilon)$ \; 
$A \gets \Var_{B}(p)$ \; 
diagonalize $A$ as $U^TDU$ with $D=\diag(\lambda_1,\dots,\lambda_n), \lambda_1 \ge \dots \geq \lambda_n, \allowbreak U=[u_1,\dots,u_n]$ \; 
$K \gets \#\{i \in \{1,\dots,n\} \mid \lambda_i \ge \delta \}$ \;
\Return $K, (u_1,\dots,u_K,u_{K+1})$} 
\end{algorithm}

\begin{algorithm}[!ht]
\caption{Compute curvature}
 \label{algorithm:curvature}
 \SetKwInOut{Input}{input}
 \SetKwInOut{Output}{output}
 \SetFuncSty{textrm}
 \SetCommentSty{textrm}
 \SetKwFunction{ComputeCurvature}{{\scshape ComputeCurvature}}
 \SetKwProg{myfunc}{}{}{}
\Input{$B$: point set, $p \in B$: point, $u_1, \dots, u_K, u_{K+1}$: orthonormal vectors}
\Output{$\curv(B, p; u_1, \dots, u_K, u_{K+1})$: real number}
\myfunc{\ComputeCurvature{$X, p;u_1, \dots, u_K, u_{K+1}$}}{
\ForEach{$q \in B, i=1, \dots, k+1$}{$x_i(q) \gets (q-p)\cdot u_i $\;}
solve $\frac{\partial}{\partial a_{i,j}} \left( \sum_{q \in B} \left( \frac{1}{2} \sum_{i,j}a_{i,j} x_i(q) x_j(q) -x_{K+1}(q) \right)^2 \right)=0$ with $a_{i,j}=a_{j,i}$ \;
\Return $\det 
\begin{bmatrix}
      a_{1,1} & \cdots & a_{1,K} \\
      \vdots & \ddots & \vdots \\
      a_{K,1} & \cdots & a_{K,K}
\end{bmatrix}$}
\end{algorithm}

\begin{algorithm}[!ht]
\caption{Compute dimension and curvature}
 \label{algorithm:dim-curv}
 \SetKwInOut{Input}{input}
 \SetKwInOut{Output}{output}
 \SetFuncSty{textrm}
 \SetCommentSty{textrm}
 \SetKwFunction{ComputeDimensionCurvature}{{\scshape ComputeDimensionCurvature}}
 \SetKwProg{myfunc}{}{}{}
\Input{$X$: point set, $p \in X$: point, $\varepsilon, \delta$: positive numbers}
\Output{$\dim(X, p;\varepsilon,\delta), \curv(X, p;\varepsilon,\delta)$: non-negative integer and real number}
\myfunc{\ComputeDimensionCurvature{$X,p;\varepsilon,\delta$}}{
$K, (u_1, \dots, u_K, u_{K+1}) \gets \textsc{ComputeDimension}(X, p;\varepsilon, \delta)$ \; 
$B \gets X \cap B(p; \varepsilon)$ \;
$c \gets \textsc{ComputeCurvature}(B, p; u_1, \dots, u_K, u_{K+1})$ \; 
\Return $K, c$}
\end{algorithm}

\subsection{Clustering based on curvatures}\label{subsec:clustering}

Using the curvature values, we also propose a clustering algorithm. 
First we cluster the obtained curvature values, then add the cluster information as the $(n+1)$th coordinate, and finally apply single-linkage clustering algorithm. 
More precisely, our clustering algorithm is the following:
\begin{enumerate}
    \item set a scaling parameter $t>0$ and a threshold $d, d'>0$;
    \item compute the curvature value $c(p)$ for $p \in X$ depending on thresholds $\varepsilon, \delta>0$ and set 
    \begin{align}
        a(p) :=
        \begin{cases}
            -t & (c(p)<-d) \\
            0 & (|c(p)| \le d) \\
            t & (c(p)>d)
        \end{cases};
    \end{align}
    \item embed $X$ into $\bR^{n+1}$ by $p \mapsto (p,a(p))$ and write $X'$ for the image; 
    \item connect every two point in $X'$ with distance less than $d'$ by an edge and output the resultant connected components. 
\end{enumerate}
If we vary the threshold $d'$ in the step~(iv), this corresponds to computing hierarchical clustering or 0-th persistence homology in persistent homology theory.
When we use a fixed threshold $d'$, one need to be set $d'$ less than $t$ so that the curvature coordinate effectively works. Namely, we want to distinguish two points $p$ and $p+\varepsilon'$  that are very near but have quite different curvature value, hence the threshold should be smaller than $\lim_{\varepsilon' \to 0} d((p,0), (p+\varepsilon', t)) = t$. 
In our experiments presented in \cref{sec:exp}, we set the distance threshold $d'$ to be $t/2$.
We remark that in the step~(iv) we can use other methods than single-linkage clustering, such as $k$-means or mean-shift clustering. 
Applying a clustering algorithm to the embedded point set instead of the original $X$ enables us to separate convex and concave parts in the point set.
The experiment with artificial and real-world data sets will be shown in the next section. 


\section{Experiments}\label{sec:exp}

In this section, we showcase our method to compute curvature-like quantity from artificial and real-world point sets.
Although our method can be applied to point sets in any dimensional Euclidean space, we restrict ourselves to point sets in $3$-dimensional space for visualization.

In practical use, the density around a point in a given point set $X$ is different from point to point in general. 
Hence the number of points in the intersection $X \cap B(p;\varepsilon)$ heavily depends on the point $p$ if we use a fixed $\varepsilon$.
For this reason, we adjust $\varepsilon$ for each point, from the viewpoint of density.
More precisely, for the radius of the ball centered at $p$ we use $\varepsilon(p)$ computed as follows.
\begin{enumerate}
    \item Fix a positive number $\eta>0$ and set $r = 1/10 \times \text{(the diameter of $X$)}$. 
    \item For $p \in X$, set $N(p):=\#(X \cap B(p;r))$, the number of points in $X \cap B(p;r)$.
    \item Define $\varepsilon(p):=2\eta/N(p)$ for each $p \in X$.
\end{enumerate}
In all the experiment in this section, we set $\eta$ and compute $\epsilon$ as described above.
Moreover, the orientation was always taken as explained in \cref{subsec:curv-alg}.

\subsection{Artificial data sets}

In this subsection, we show the result of computation of the curvatures for artificial point sets.
\subsubsection{Simple examples}
First, we observe the behavior of our curvature value for two simple point sets.
We generated two point sets consisting of $3000$ points on the hypersurfaces $z=-x^2-y^2$ and $z=x^2-y^2$, respectively. 
The parameter $\varepsilon$ was set as explained above with $\eta$ equal to $3$ times the diameter of the point set and $\delta=0.001$. 
The results are shown in \cref{fig:neg_pos}. 

\begin{figure}[!ht]
    \centering
    \includegraphics[width=120mm]{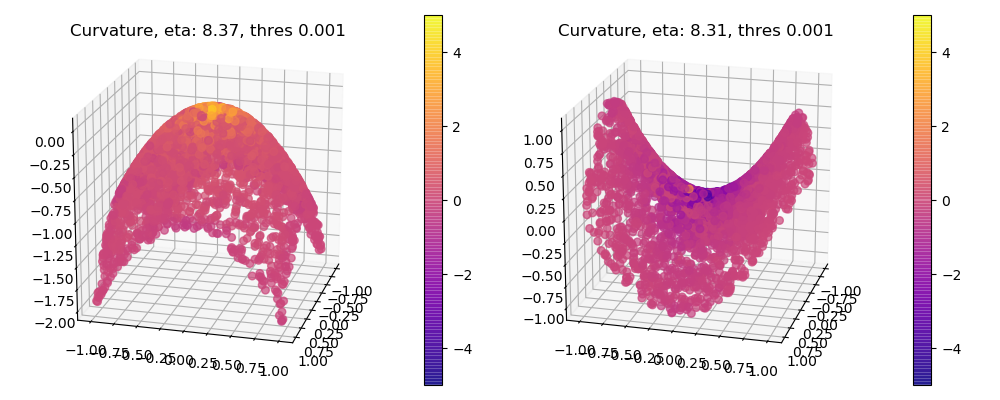}
    \caption{The curvature values for the surfaces $z=-x^2-y^2$ and $z=x^2-y^2$.}
    \label{fig:neg_pos}
\end{figure}

In these results, we can see that for the surface $z=-x^2-y^2$ (left) the curvature values are positive while for $z=x^2-y^2$ (right) the curvature values are negative, which coincide with the usual Gaussian curvatures $4/(1+4x^2+4y^2)^2$ and $-4/(1+4x^2+4y^2)^2$, respectively. 
As explained in \cref{subsec:gauss}, the curvature values are independent from the choice of frames since they are 2-dimensional.

\subsubsection{Topologically similar examples distinguished by curvature values }\label{topex}

Next we showcase two point sets that are difficult to distinguish by persistent homology but can be easily distinguished by our method. 
Roughly speaking, persistent homology captures the topology of point set $X$ in $\bR^n$ and summarizes the information as a persistence diagram, which is a multiset in $(\bR \cup \{\infty\})^2$. 
Here, we briefly explain what they are (see, for example \cite{edelsbrunner2010computational}, for more details).
The idea is to consider the homology of the union $X_r = \bigcup_{p \in X} \overline{B}(p;r)$ for $r \ge 0$, where $\overline{B}(p;r)$ denotes the closed ball of radius $r$ centered $p$, and the family $(H_n(X_r), {i_{r,s}}_*)_{r \le s \in \bR_{\ge 0}}$ of homology and homomorphisms induced by inclusions $i_{r,s} \colon X_r \hookrightarrow X_s$. 
The family is called the \emph{persistent homology}. 
Any homology class $\alpha$, which represents a connected component, a loop, or a cavity, etc., is generated at $r=d_\alpha$ and vanishes at $r=d_\alpha$, where $d_\alpha=\infty$ if $\alpha$ does not vanish. 
The collection $\{(b_\alpha, d_\alpha)\}_\alpha$ for all homology classes is called the \emph{persistence diagram}. 
A homology class $\alpha$ can be seen as a topological noise if $d_\alpha-b_\alpha$ is small, which means it vanishes just after the generation.
Thus, in a persistence diagram, a point near the diagonal is regarded as a noise and can be ignored.

In the experiment, we used the two point sets displayed in the upper row of \cref{fig:pd_histogram}.
The first one is a finite set of $3000$ points on the union of $S^1 \times [0,1]$ and two discs (upper left), where we randomly sampled $1500$ point from $S^1 \times [0,1]$ and $750$ points from discs, respectively.
The second one is a finite set of $3000$ points from the sphere with radius $0.5$ (upper right).
We computed the persistence homology of each point set and output the corresponding persistence diagram, which is shown in the middle row of \cref{fig:pd_histogram}.
Here, we used alpha complex for reducing the computation complexity, with the use of \texttt{GUDHI} library\footnote{\url{https://gudhi.inria.fr}} in Python. 
We also computed the curvature values at each point in the point sets with $\eta=3 \times \text{(the diameter of the point set)}, \delta=0.001$ and made the histogram of the values, which is shown in the third row of \cref{fig:pd_histogram}.

\begin{figure}[!ht]
    \centering
    \includegraphics[width=115mm]{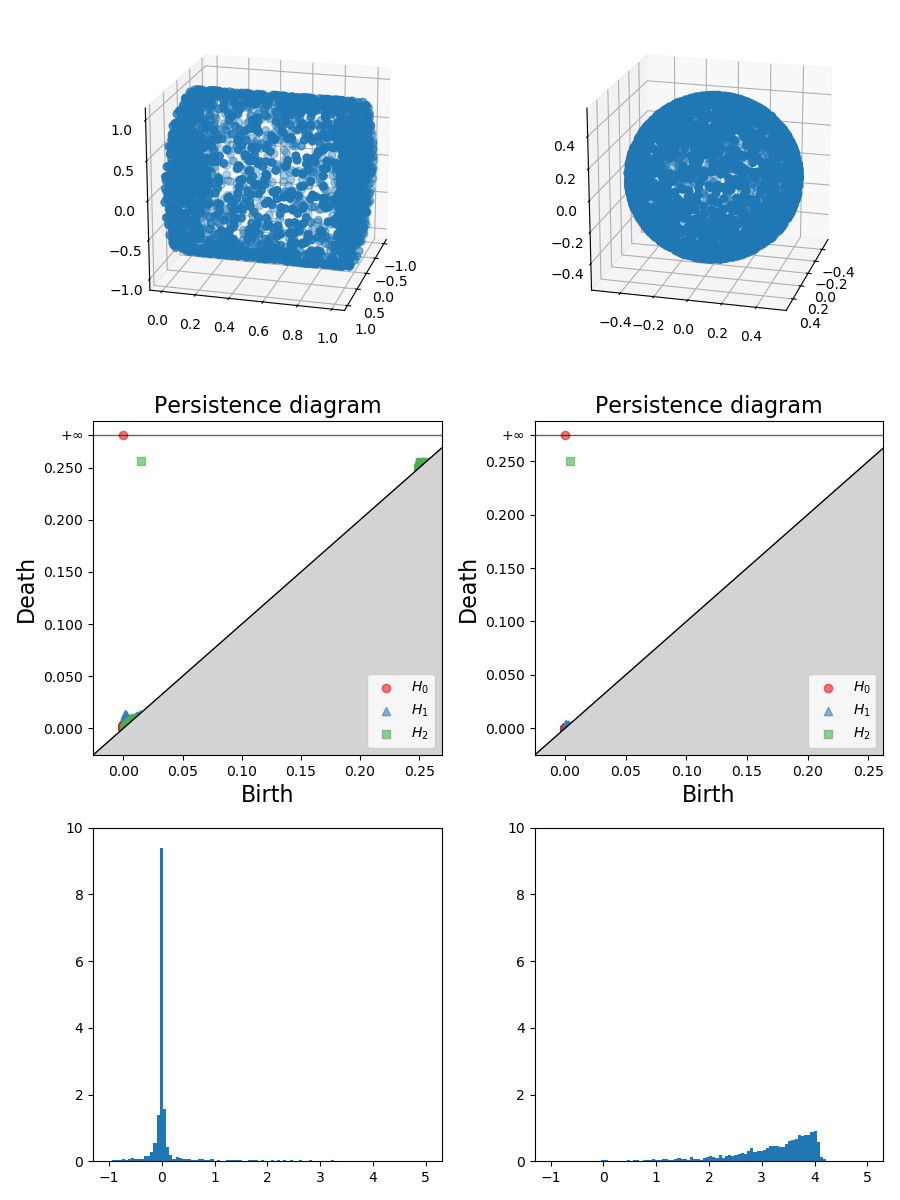}
    \caption{The two point sets that have similar persistence diagrams and different histograms of curvature values. The major difference between the diagrams appear only near the diagonal and can be ignored from the view of persistent homology. While the two diagrams are similar apart from a neighborhood of the diagonal, the histograms of curvature values are quite different.}
    \label{fig:pd_histogram}
\end{figure}

In these results, we can see the two point sets have similar persistence diagrams while they have quite different histograms of curvature values. 
When the topological feature of data sets is not sufficient and more geometric information is needed for a task in question, our curvature values could be helpful in practical situation.

\subsubsection{More examples and clustering}

Next, we observe the behavior of curvature values depending on the choices of $\varepsilon$.
For this purpose, we generated two artificial point sets as follows.
The first one is the union of $S^1 \times [0,1]$ and two discs same as in \cref{fig:pd_histogram} in \cref{topex}. 
Here we randomly sample $1500$ point from $S^1 \times [0,1]$ and $750$ points from each disc.
The second one is a union of $S^1 \times [0,1]$ and two hemi-ellipsoids (\cref{fig:column_spheres}). 
Similarly to the above, we randomly sample $1500$ point from $S^1 \times [0,1]$ and $750$ points from each hemi-ellipsoid.
We applied our algorithm to these point clouds consisting of $3000$ points.
Here the threshold $\delta$ was set to be $0.001$ and 
 where $\eta$ was set to be
\begin{align}\label{eq:eta}
    k \times \text{(the diameter of a given point cloud)} \quad (k=1,2,3,4).
\end{align}
The results are shown in \cref{fig:column_discs,fig:column_spheres}. 
The left columns show the estimated dimensions and the right ones indicate the curvatures. 

In \cref{fig:column_discs}, apart from points near the edges, the dimension is $2$ for most of the points, which coincides with the true dimension. 
Moreover, most of the points on the flat parts have curvature value close to $0$ for different $\eta$.
On the other hand, a point on the edges has non-zero value because the fitted quadratic surface has high curvature at the point.
For a larger $\eta$, more points are used for the calculation of the principal component analysis and the eigenvalues of the variance matrix tends to be larger near the edges. 
This explains why enlarging $\eta$ increases the number of points whose curvature vales are $3$ near the edges. 
In \cref{fig:column_spheres}, we can observe that points near the two poles have non-zero curvature values while those on the other parts have value close to $0$. 
For this point sets, enlarging $\eta$ increases the number of points with estimated dimension $3$.

\begin{figure}[!ht]
    \centering
    \includegraphics[width=0.95\textwidth]{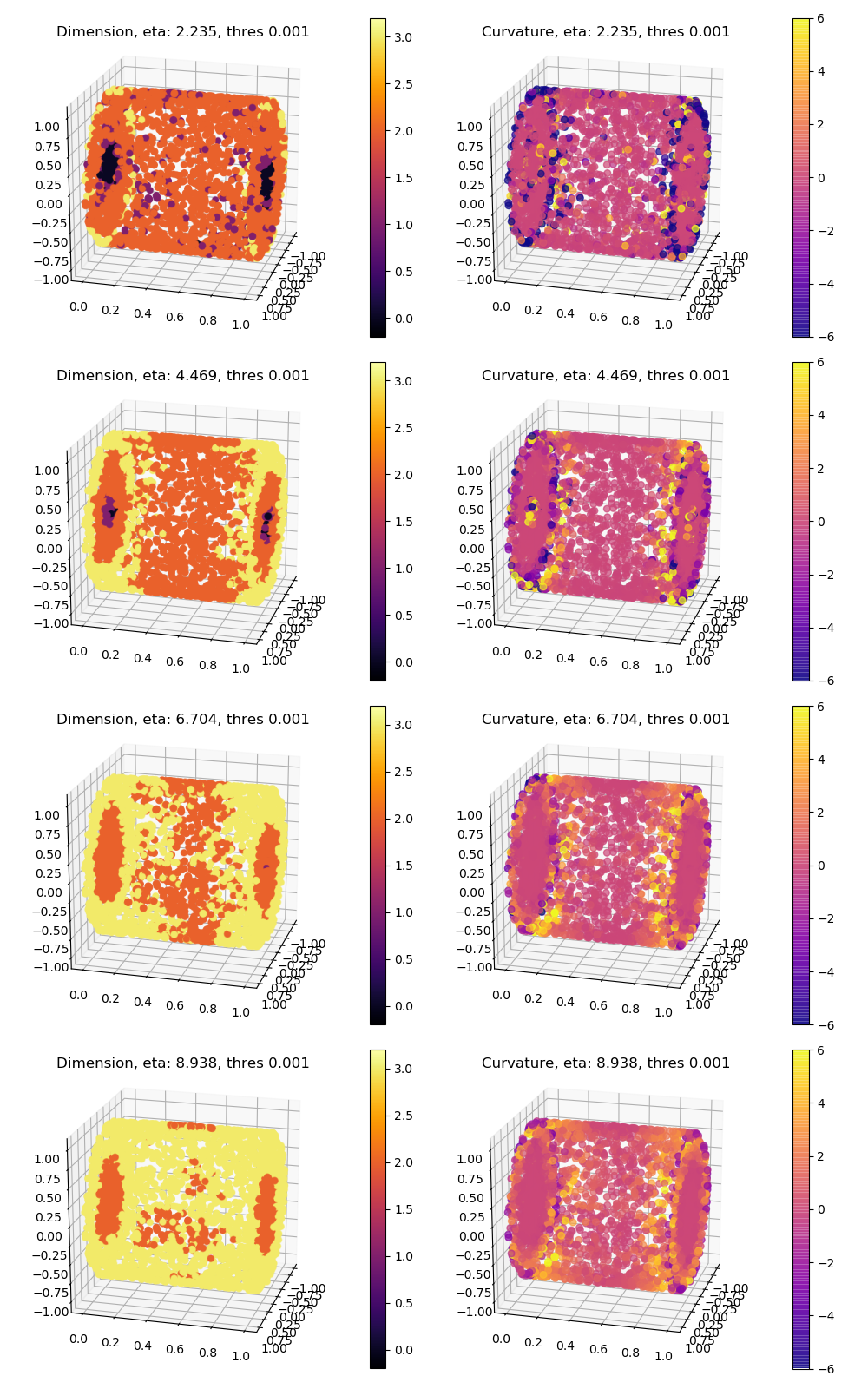}
    \caption{The estimated dimensions (left) and the curvature values (right) of $S^1 \times [0,1]$ union two flat discs.}
    \label{fig:column_discs}
\end{figure}

\begin{figure}[!ht]
    \centering
    \includegraphics[width=0.95\textwidth]{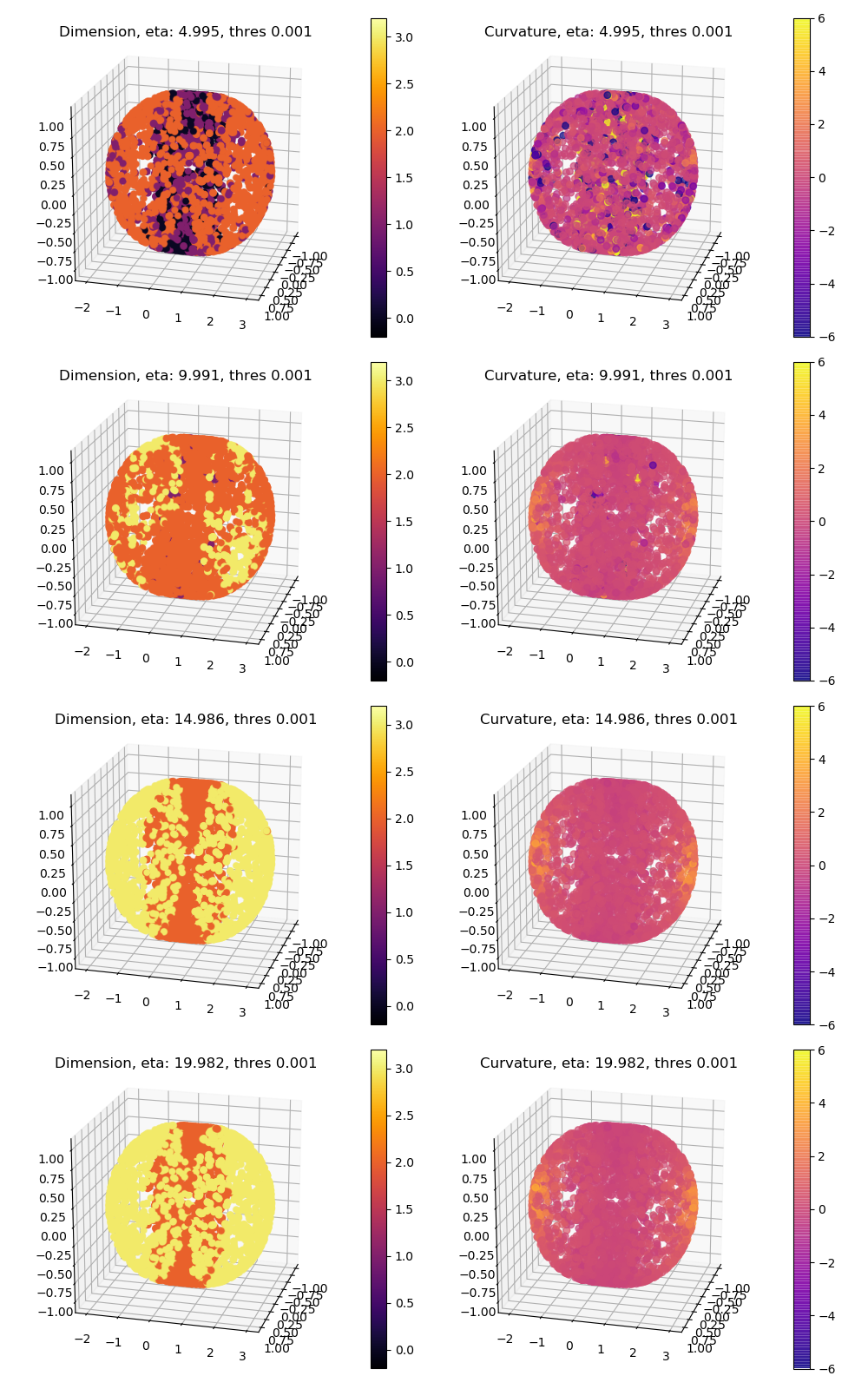}
    \caption{The estimated dimensions (left) and the curvature values (right) of $S^1 \times [0,1]$ union two hemi-ellipsoids.}
    \label{fig:column_spheres}
\end{figure}

We also applied our clustering algorithm explained in \cref{subsec:clustering} to these point sets.
We set the scaling parameter $t$ to be $4$, $d'=t/2$, and the threshold $d$ to be $0.5$. 
The results for clustering algorithm are showed in \cref{fig:clustering} depending on different choices of $\varepsilon$. 
Since the curvature values were added as the fourth coordinates, the clustering algorithm separated points with different curvature values as well as coordinate values in $\bR^3$.
For $S^1 \times [0,1]$ union two flat discs, the algorithm clustered the points sets into three parts, that is, the edges with high curvatures, the flat discs union the middle part of the column with curvature close to $0$, and the rest. 
Adjusting the scaling parameter $t$ and $d'$, one could also cluster the point sets into neighborhoods of two edges on both sides and the rest. 
For $S^1 \times [0,1]$ union two hemi-ellipsoids, the algorithm clustered the point sets into three parts, two regions near the poles and the rest, as expected. 
This is because the regions near the poles had higher curvature values and these regions were far enough in $\bR^3$.

\begin{figure}[H]
    \centering
    \includegraphics[width=130mm]{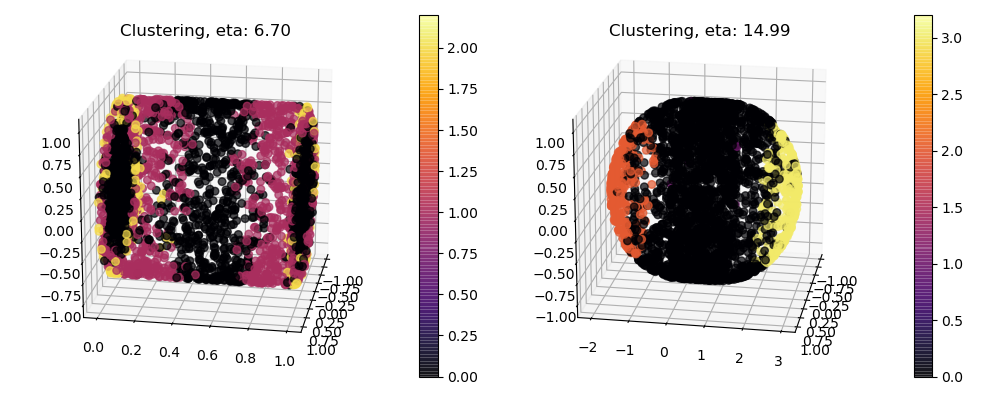}
    \caption{The clustering results of the artificial data sets.}
    \label{fig:clustering}
\end{figure}

\subsubsection{Noisy datasets}\label{subsec:noise}

Here, we show the validity of our construction of curvature function by checking \eqref{eq:lln}, that is, mean of the curvature coincides with the curvature of the mean manifold. We consider the behavior of our algorithm on noisy data.

We generate point sets by the following procedure.
First, we generate randomly 1000 points $a_1, \dots, a_{1000}$ on the plane $\bR^2$.
Next, we consider a set of random variables $(X_1, \dots, X_{1000})$ that follows $N(0,C)$, the normal distribution with mean $0$ and covariance matrix $C=(C_{ij})$ defined by 
\begin{equation}
    C_{ij} = \exp(-\|a_i-a_j\|_2^2).
\end{equation}
Consider a surface $z=f(x,y)$ with $f \equiv 0$ or $f(x,y)=\pm \sqrt{1-x^2-y^2}$.
We generate 1000 points by setting $p_i := f(a_i) + (\sigma X_i) n_i$, where $\sigma>0$ and $n_i$ is a unit normal vector of the surface at $f(a_i)$. 
The first distribution is same as the one explained in the fourth paragraph of \cref{sec:framework}. 
As mentioned in \cref{rem:nsw}, this modeling can also be seen as the model proposed by Niyogi--Smale--Weingberger~\cite{NSW}.

We compute the curvature values for the point set $X=\{p_1, \dots, p_{1000} \}$ for 50~times with $\sigma=0.1, \eta=1, \delta=0.005$.
Then we compute the mean of the curvature values for each $i$, which are displayed in \cref{fig:noisy}. 
We can see that mean of the curvature is close to the curvature of the mean manifold. 

\begin{figure}[H]
    \centering
    \includegraphics[width=0.45\textwidth]{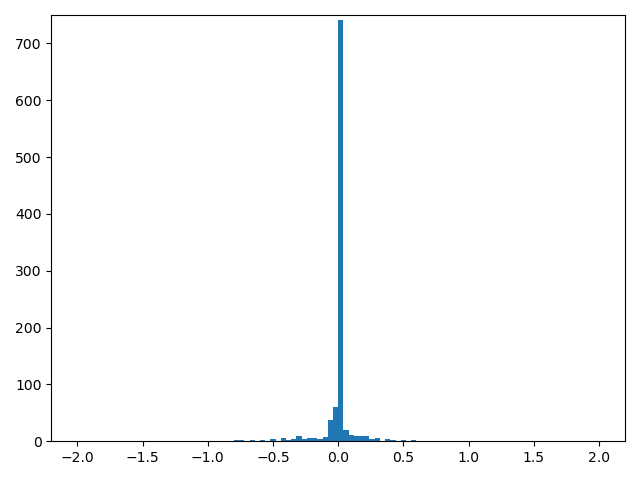}
    \includegraphics[width=0.45\textwidth]{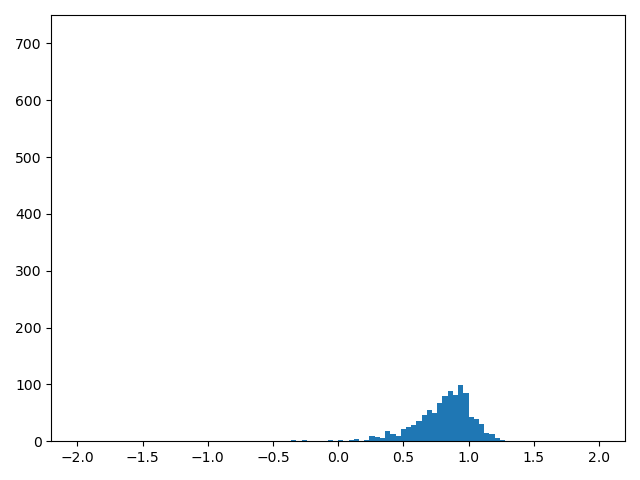}
    \caption{The distributions of the curvature values of noisy planes (left) and noisy spheres (right) with $\sigma=0.1$. We generated point sets 50~times and computed the averages.}
    \label{fig:noisy}
\end{figure}

Now we conjecture that for stochastic embeddings with suitable kernel function and suitable curvature function, the diagram \eqref{diagram1} commutes. It has already verified for the distribution with $f(x, y) \equiv 0$ in the fourth paragraph of \cref{sec:framework}.


\subsection{Real-world data sets}

Now we show the results for real-world data sets. 
For these experiments, we used \texttt{G-PCD: Geometry Point Cloud Dataset}\footnote{\url{https://www.epfl.ch/labs/mmspg/downloads/geometry-point-cloud-dataset/}}, from which we pick up Stanford Bunny and Stanford Dragon point sets.
We randomly took $3000$ points from each point set and applied our algorithm, setting $\delta$ to be $0.001$ and $\eta$ varies as in \eqref{eq:eta}.
The results are shown in \cref{fig:bunny,fig:dragon}.

In \cref{fig:bunny}, we can find that the estimated dimension fits well the true dimension $2$ for most of the points when $\eta = 3.816$ or $5.088$.
For $\eta=1.272$, only few points were used for the calculation of the principal component analysis and the algorithm could not estimate the covariance matrix well at each point. 
As $\eta$ gets larger, the estimated dimension at each point tends to be bigger. 
This is again because more points are used for the principal component analysis and the eigenvalues of the variance matrix tends to be larger near the edges. 
Indeed, the points near the base and the ears have estimated dimension $3$ for $\eta=5.088$.
For curvature values, most of the points on the flat body parts have values close to $0$ for each $\eta$, thus thus can be seen flat with our curvature. 
On the other hand, when $\eta = 3.816$ and $5.088$, the curvature values at convex or concave part, e.g., the ears and the legs, have large absolute values. 
\cref{fig:dragon} indicates a similar result for the dragon point sets. 
In this case, the high convexity and concavity causes the large absolute values of curvature for most of the points when $\eta=3.070$ or $4.093$.

\begin{figure}[ht]
    \centering
    \includegraphics[width=125mm]{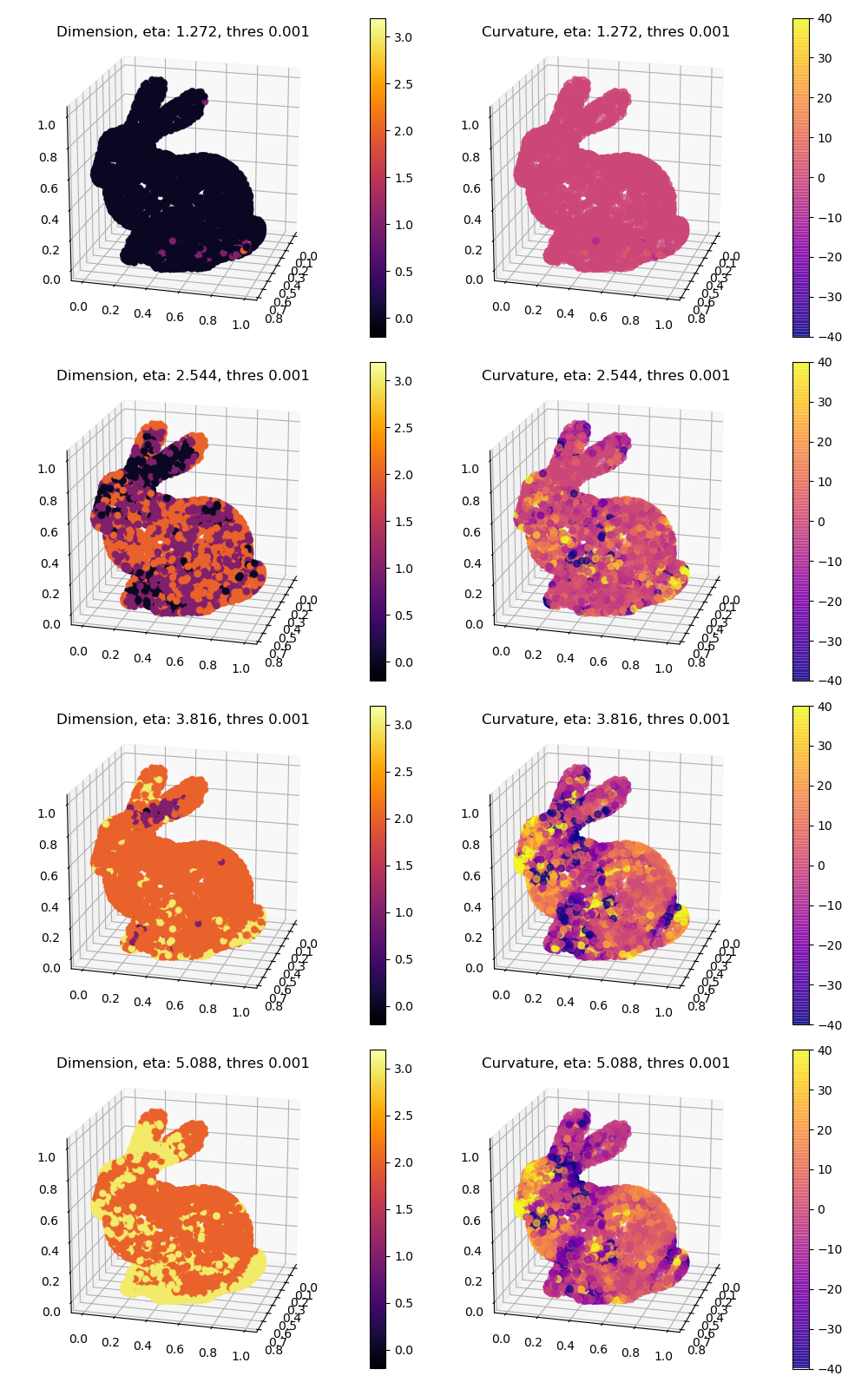}
    \caption{The estimated dimensions (left) and the curvature values (right) of the Stanford bunny data set.}
    \label{fig:bunny}
\end{figure}

\begin{figure}[ht]
    \centering
    \includegraphics[width=125mm]{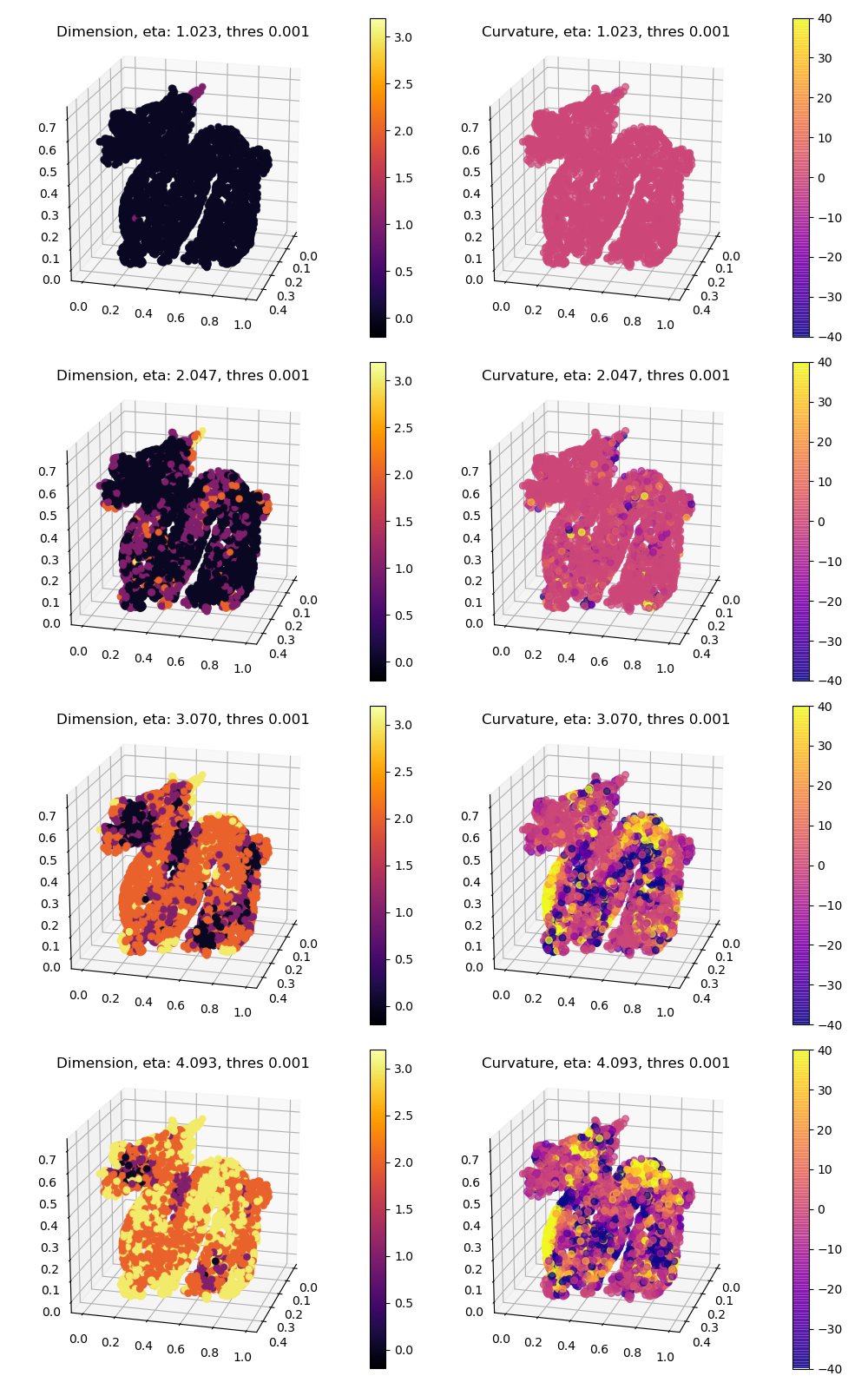}
    \caption{The estimated dimensions (left) and the curvature values (right) of the Stanford dragon data set.}
    \label{fig:dragon}
\end{figure}

Using the curvature values, we also applied the clustering algorithm presented in \cref{subsec:clustering} to these point sets. 
Here, we set the scaling parameter $t$ to be $4$, $d'=t/2$, and the threshold $d$ to be $0.5$. 
The results are shown in \cref{fig:clustering_stanford}.
The algorithm clusters the bunny point set into three parts, the flat body part, the legs plus the ears that have large absolute value of curvature, and outliers. 
For the dragon point set, the result is not satisfactory as the bunny set. 
The method clusters the point set into three parts, depending on curvature value is sufficiently positive, close to 0, and sufficiently negative. 
In this case, one of the cluster can be seen as the flat body parts while the meaning of the others is not so clear.

\begin{figure}[H]
    \centering
    \includegraphics[width=0.9\textwidth]{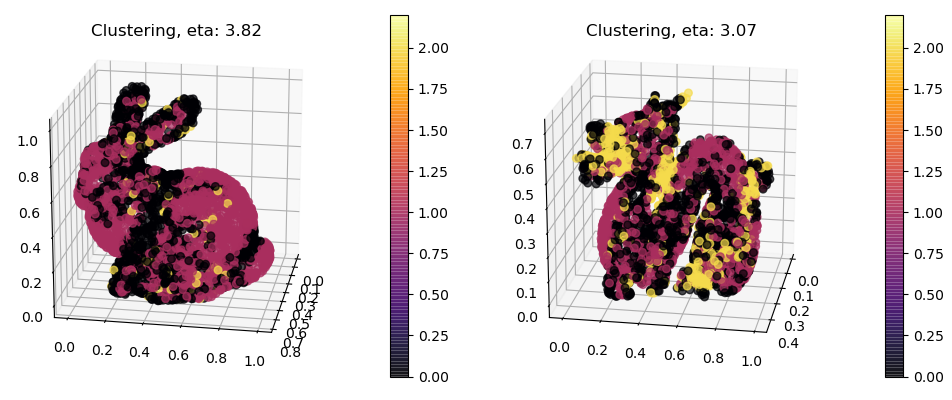}
    \caption{The clustering result of the real-world data sets.}
    \label{fig:clustering_stanford}
\end{figure}

\section{Conclusion}

Curvature is a essential quantity for studying spaces and would be helpful for studying point sets in data analysis. 
In this article, we have proposed a mathematical framework that indicates how the curvature of data should behave. Then we have given a construction of curvature function for data, and have shown its validity by numerical experiments. Here we conjecture that our framework is appropriate in the sense that the diagram \eqref{diagram1} commutes in suitable situations. We also have proposed a method to compute the dimensions and the curvatures of point sets with respect to our definitions.
We have also given some experimental results which supports effectiveness of our method.
Fixing the parameter $\varepsilon$ in our algorithm would take unexpected values due to noises, thus we proposed computing the curvature values as functions of $\varepsilon$.
We expect that such function-like treatment would give robustness against noise, and we would build mathematical foundation for them in future work.

\clearpage

\bibliographystyle{alpha}
\bibliography{references}

\end{document}